\documentclass[11pt]{article}
\usepackage{fullpage}
\usepackage{tikz}
\usetikzlibrary{arrows.meta,positioning,matrix,decorations.pathreplacing,decorations.markings,calligraphy,backgrounds}
\usetikzlibrary{calc}

\usepackage{amsmath}
\usepackage{amssymb}
\usepackage{enumitem}
\usepackage{amsthm}
\usepackage[most]{tcolorbox}
\usepackage{subfiles}
\usepackage{algorithm}
\usepackage{algpseudocode}
\usepackage{mathtools}
\usepackage{bm}

\usepackage{abraces}
\usepackage[dvipsnames]{xcolor}
\usepackage{nicematrix}
\usepackage[backref=page]{hyperref}
\hypersetup{
    colorlinks=true,
    linkcolor=blue,
    filecolor=magenta,      
    urlcolor=cyan,
}
\usepackage{thmtools}
\usepackage{thm-restate}
\usepackage{cleveref}
\usepackage{dsfont}
\usepackage{blkarray}
\usepackage{booktabs}
\usepackage{caption}

\newtheorem{theorem}{Theorem}[section]
\newtheorem{corollary}[theorem]{Corollary}

\newtheorem{lemma}[theorem]{Lemma}

\newtheorem{defi}[theorem]{Definition}

\newtheorem{conjecture}[theorem]{Conjecture}

\newtheorem{fact}[theorem]{Fact}

\crefname{defi}{definition}{definitions}
\Crefname{defi}{Definition}{Definitions}

\newcommand{\ket}[1]{\ensuremath{\left|#1\right\rangle}}

\newcommand{\Q}{\ensuremath{\mathsf{Q}}}

\newcommand{\zone}{\{0,1\}}

\newcommand{\cbra}[1]{\left\{#1\right\}}

\newcommand{\OR}{\mathsf{OR}}

\newcommand{\PrMaj}{\mathsf{PrMaj}}
\newcommand{\wt}{\mathsf{wt}}

\newcommand{\umod}{\mathsf{MD}}
\newcommand{\dmod}{\mathsf{MOD}}
\newcommand{\mcut}{\mathsf{MaxCut}}

\newcommand{\odeg}{\mathsf{deg}^{+}}
\newcommand{\ideg}{\mathsf{deg}^{-}}

\allowdisplaybreaks

\title{Quantum Query Complexity of Maximum Degree \\and its Generalizations}

\author{Amin Shiraz Gilani\thanks{QuICS, University of Maryland, USA {\tt asgilani@umd.edu}}
\and
Nikhil S.~Mande\thanks{University of Liverpool, UK {\tt mande@liverpool.ac.uk}}
\and 
Subhasree Patro\thanks{Eindhoven University of Technology, Netherlands {\tt patrofied@gmail.com}}
\and
Nithish Raja\thanks{Eindhoven University of Technology, Netherlands {\tt n.r.raja@tue.nl}}
}
\date{}

\begin{document}

\maketitle

\begin{abstract}
We study the quantum query complexity of several basic graph problems. 
\begin{itemize}
    \item We prove a tight $\Theta(n^{3/2})$ bound for finding a maximum-out-degree vertex in an $n$-vertex tournament, resolving a conjecture of Mande, Paraashar, and Saurabh [FSTTCS'23]. 
    \item Using the same proof idea, we give a direct proof that finding a maximum-degree vertex in an $n$-vertex undirected graph requires $\Theta(n^{3/2})$ quantum queries, a bound implicit in the work of Apers and Lee [CCC'21]. Both lower bounds above hold even for computing only the maximum degree or maximum out-degree.
    \item We also study Maximum $(k,n-k)$-Cut, which asks for the largest cut with one side of size $k$. For every constant $k\geq 1$, we give an $O(n^{2-1/(k+1)})$-query quantum algorithm.
    \item We establish a reduction to Maximum $(k,n-k)$-Cut involving $k$-Graph Collision, which asks whether the marked vertices of a known graph contain a $k$-clique. For $k\geq 3$, this yields unconditional lower bounds with exponents strictly larger than $3/2$, showing that this problem is strictly harder than computing maximum degree. Additionally, under a conjecture we propose for $k$-Graph Collision, the reduction yields a lower bound matching our upper bound above.
\end{itemize}

\end{abstract}

\newpage
\setcounter{page}{1}

\section{Introduction}

Graph problems provide a natural setting for studying the power and limitations of quantum query algorithms. In the adjacency-matrix model, an algorithm accesses a graph by querying whether a given pair of vertices forms an edge, and the goal is to solve a graph problem using as few queries as possible. Quantum algorithms can query this adjacency information in superposition, leading to query savings for several fundamental graph problems. For example, D\"urr et al.~\cite{Durr_Heiligman_Hoyer_Mhalla_2006} established quantum speedups for connectivity and minimum spanning tree. These results form part of a broader study of quantum query complexity for graph problems~\cite{Buhrman_Cleve_de_Wolf_Zalka_1999,Berzina_Dubrovsky_Freivalds_Lace_Scegulnaja_2004,Jeffery_Kimmel_2017,Apers_Lee_21,Terao_Mori_2024,BCGKPW_2024,Cornelissen_Gilani_Patro_2026}.

Subgraph detection has been another central theme in this line of work. The problem of deciding whether an undirected graph contains a triangle has received particular attention, motivating a sequence of quantum algorithms with improved query complexity~\cite{Magniez_Santha_Szegedy_2007,Lee_Magniez_Santha_2013,Jeffery_Kothari_Magniez_2013,Le_Gall_2014,Le_Gall_Nakajima_2017,Carette_Lauriere_Magniez_2020}. More generally, quantum algorithms have been studied for detecting other fixed subgraphs and for deciding minor-closed graph properties~\cite{Zhu_2012,Childs_Kothari_2012, Kulkarni_Podder_2016,Cornelissen_Gilani_Patro_2026, Gilani_Zhou_2026}.

Graph optimization problems have been extensively studied in the standard model of computation and also provide a natural setting for studying the power of quantum query algorithms. A fundamental example is \textsc{MaxCut}: given an undirected graph $G=(V,E)$, find a set $S\subseteq V$ maximizing the number of edges between $S$ and $V\setminus S$. This problem is NP-hard~\cite{Karp72} and has been studied extensively from the perspective of approximation algorithms~\cite{GW95}. A natural variant, Maximum $(k,n-k)$-Cut, requires $S$ to have size exactly $k$ and is W[1]-hard when parameterized by $k$~\cite{Cai_2008}. Such problems motivate studying whether quantum algorithms can reduce the number of queries needed to find an optimal solution, and how this complexity depends on the parameter $k$. Recent work of Terao and Mori~\cite{Terao_Mori_2024} studies parameterized quantum query complexity for graph problems, obtaining upper and lower bounds for parameterized versions of Vertex Cover and Matching.

Despite this progress, the quantum query complexity of some basic graph problems remains unresolved. In this work, we consider the following questions, given query access to the adjacency matrix of an $n$-vertex graph:
\begin{enumerate}
\item What is the quantum query complexity of finding a maximum-degree vertex in an undirected graph, or computing its maximum degree?
\item More generally, what is the quantum query complexity of Maximum $(k,n-k)$-Cut? The case $k=1$ is precisely the problem of computing the maximum degree.
\item What is the quantum query complexity of finding a maximum-out-degree vertex in a tournament, or computing its maximum out-degree?
\end{enumerate}

For the first and third questions, a reasonably straightforward quantum algorithm gives an $O(n^{3/2})$ upper bound: the degree or out-degree of a specified vertex can be computed using $n-1$ queries, and quantum maximum finding~\cite{DH96} uses $O(\sqrt{n})$ such evaluations to find a maximizing vertex. Its degree or out-degree can then be computed with another $n-1$ queries. 

Before our work, the following was the status of lower bounds on the above three problems.
\begin{enumerate}
    \item For undirected graphs, a matching lower bound is implicit in the work of Apers and Lee~\cite[Theorem 37]{Apers_Lee_21}. Their proof (of a query lower bound for minimum cut) encodes a composed counting problem into weighted vertex degrees. A simple query-preserving transformation of their hard instances yields the $\Omega(n^{3/2})$ lower bound for unweighted maximum degree. We give a direct proof using a reduction from $\OR\circ\PrMaj$.
    \item To the best of our knowledge, the second question has not previously been studied in the quantum query model.
    \item For tournaments, Mande, Paraashar, and Saurabh~\cite{MPS23} explicitly asked whether this upper bound is tight.
\end{enumerate}

We establish tight quantum query bounds for finding maximum degree and maximum-out degree of an undirected graph and a tournament, respectively. For Maximum $(k,n-k)$-Cut, we give upper and unconditional lower bounds, and prove that the upper bound is tight assuming a conjecture on $k$-Graph Collision.

\subsection{Our results}

Let $\umod_n$ and $\dmod_n$ denote the tasks of finding a maximum-degree vertex in an $n$-vertex undirected graph and a maximum-out-degree vertex in an $n$-vertex tournament, respectively, and let $\Q(\cdot)$ denote bounded-error quantum query complexity. We refer the reader to \Cref{sec:prelims} for formal definitions.

\begin{restatable}{theorem}{thmmaxdegree}\label{thm:max-degree}
    The bounded-error quantum query complexity of finding a maximum-degree vertex in an $n$-vertex undirected graph is $\Theta(n^{3/2})$, i.e., $\Q(\umod_n) = \Theta(n^{3/2})$.
\end{restatable}

\begin{restatable}{theorem}{thmmod}\label{thm:mod}
    The bounded-error quantum query complexity of finding a maximum-out-degree vertex in an $n$-vertex tournament is $\Theta(n^{3/2})$, i.e., $\Q(\dmod_n) = \Theta(n^{3/2})$.
\end{restatable}

Our lower bounds also hold for the weaker tasks of computing only the maximum degree or maximum out-degree. The matching upper bounds follow from quantum maximum finding~\cite{DH96}, as described above.

We next turn to Maximum $(k,n-k)$-Cut. Our lower bound is expressed in terms of the \emph{$k$-Graph Collision} problem, described below. Fix a known graph $G=(V,E)$, and suppose that an input $x\in\{0,1\}^V$ labels each vertex as marked or unmarked. The problem $\mathsf{GC}^G_k$ asks whether the marked vertices contain a $k$-clique. When $k=2$, this is the standard Graph Collision problem~\cite[Sec~4.2]{Magniez_Santha_Szegedy_2007}.

We introduce this generalization concurrently with Gilani and Zhou~\cite{Gilani_Zhou_2026}, who study it under the name \emph{clique collision}.\footnote{A version of this problem appears implicitly in the subgraph-finding algorithm of Magniez, Santha, and Szegedy~\cite[Section~4.4]{Magniez_Santha_Szegedy_2007}. Specializing their construction to the pattern $H=K_{k+1}$, fix a vertex $v$ outside a set $U$ whose induced graph is known, and mark each $u\in U$ if it is adjacent to $v$. Their checking procedure asks whether $U$ contains a $k$-clique all of whose vertices are marked, which is precisely our $k$-Graph Collision problem on the known graph induced by $U$, with queries to the marking bits implemented by queries to edges incident to $v$.} 

\begin{restatable}{theorem}{thmmaxcut}\label{thm:max-k-cut}
For every constant $k\geq 1$, $\Omega(n^{3/2}) = \Q(\mcut_{k,n})
    = O\!\left(n^{2-\frac{1}{k+1}}\right)$.
Moreover, for every even positive integer $N$, integer $1 \leq k \leq N$, and $N$-vertex graph $G$, 
\[
\Q(\mcut_{k,3N})
    =\Omega\!\left(N\cdot \Q(\mathsf{GC}^G_k)\right).
\]
\end{restatable}

Combining the reduction in \Cref{thm:max-k-cut} with the lower bounds for clique collision of Gilani and Zhou~\cite[Theorem~1.2]{Gilani_Zhou_2026} gives stronger unconditional lower bounds for every fixed $k\geq 3$.

\begin{restatable}{corollary}{corunconditionalmaxcut}\label{cor:max-k-cut-unconditional}
For every fixed $k\geq 3$,
\begin{equation*}
    \Q(\mcut_{k,n}) = \Omega\left(n^{1+\frac{\ell(2k-\ell+1)}{2k(\ell+1)}-o(1)}\right), \qquad \textnormal{where } \ell=\left\lceil\frac{\sqrt{8k+9}-3}{2}\right\rceil.
\end{equation*}
\end{restatable}
In particular, for $k=3,4,5$, this yields lower bounds of $n^{14/9-o(1)}$, $n^{19/12-o(1)}$, and $n^{8/5-o(1)}$, respectively.

For $k=2$, it is a standard conjecture that there exist $N$-vertex graphs $G$ for which $\Q(\mathsf{GC}^G_2)=\Omega(N^{2/3})$ (see, for example,~\cite{Cornelissen_Gilani_Patro_2026} and the references therein). More generally, a standard quantum-walk argument gives $\Q(\mathsf{GC}^G_k)=O(N^{1-1/(k+1)})$ for every constant $k\geq 2$ and every $N$-vertex graph $G$. We conjecture that this upper bound is tight in the worst case.

\begin{conjecture}\label{conj:kGC}
For every constant $k\geq 2$, there exists a family of $N$-vertex graphs $G$ such that $\Q(\mathsf{GC}^G_k)=\Omega\!\left(N^{1-\frac{1}{k+1}}\right)$.
\end{conjecture}

Combining \Cref{conj:kGC} with \Cref{thm:max-k-cut} gives the following conditional tight quantum query bound for Maximum $(k,n-k)$-Cut.
\begin{corollary}\label{cor:max-k-cut-conj}
Assuming \Cref{conj:kGC}, for every constant $k\geq 2$, we have $\Q(\mcut_{k,n})=\Theta\!\left(n^{2-\frac{1}{k+1}}\right)$.
\end{corollary}

\subsection{Our techniques}
We begin with the lower bounds in \Cref{thm:max-degree,thm:mod}, then describe the upper and lower bounds in \Cref{thm:max-k-cut}.

For the lower bounds in \Cref{thm:max-degree,thm:mod}, we reduce from the following Boolean problem. We are given $m$ strings $x^{(1)},\ldots,x^{(m)}\in\{0,1\}^m$, each promised to have Hamming weight in $\{m/2,m/2+1\}$, and must decide whether any string has Hamming weight $m/2+1$. This is precisely the composed function $\OR_m\circ\PrMaj_m$, where $\PrMaj_m$ distinguishes the two promised Hamming weights (see \Cref{def:prmaj}).

Writing $b=(\OR_m\circ\PrMaj_m)(x)$, we construct an undirected graph $G_x$ on $2m$ vertices and a tournament $T_x$ on $3m+1$ vertices satisfying
\[
\Delta(G_x)=3m/2-1+b \qquad\text{and}\qquad \Delta^+(T_x)=3m/2+b.
\]
Thus, computing either $\Delta(G_x)$ or $\Delta^+(T_x)$ determines $(\OR_m\circ\PrMaj_m)(x)$. The constructions are depicted in \Cref{fig:max-degree-new} and \Cref{fig:mod-new}, respectively.

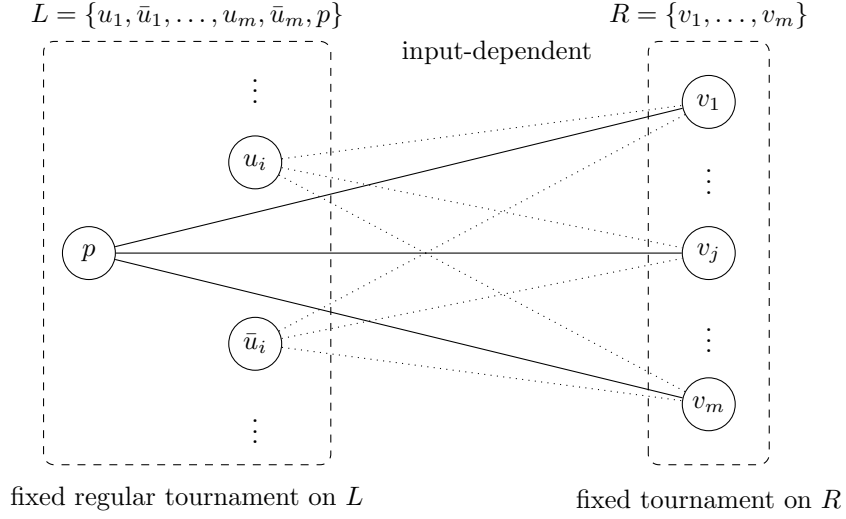
\begin{figure}[ht!]
\centering
\begin{tikzpicture}[
    scale=1,
    vertex/.style={circle,draw,minimum size=7mm,inner sep=0pt},
    every node/.style={font=\small}
]

\node[vertex] (ui)  at (0,1.2) {$u_i$};
\node[vertex] (ubi) at (0,-1.2) {$\bar u_i$};
\node[vertex] (p)   at (-2.2,0) {$p$};

\node at (0,2.25) {$\vdots$};
\node at (0,-2.25) {$\vdots$};

\draw[dashed,rounded corners]
    (-2.8,-2.8) rectangle (1.0,2.8);

\node at (-0.9,3.15)
    {$L=\{u_1,\bar u_1,\ldots,u_m,\bar u_m,p\}$};

\node[align=center] at (-0.9,-3.25)
    {fixed regular tournament on $L$};

\node[vertex] (v1) at (6,2.0) {$v_1$};
\node[vertex] (vj) at (6,0) {$v_j$};
\node[vertex] (vm) at (6,-2.0) {$v_m$};

\node at (6,1.05) {$\vdots$};
\node at (6,-1.05) {$\vdots$};

\draw[dashed,rounded corners]
    (5.2,-2.8) rectangle (6.8,2.8);

\node at (6,3.15) {$R=\{v_1,\ldots,v_m\}$};

\node[align=center] at (6,-3.25)
    {fixed tournament on $R$};

\draw[dotted] (ui) -- (v1);
\draw[dotted] (ui) -- (vj);
\draw[dotted] (ui) -- (vm);

\draw[dotted] (ubi) -- (v1);
\draw[dotted] (ubi) -- (vj);
\draw[dotted] (ubi) -- (vm);

\node at (3.2,2.65) {input-dependent};

\draw[-] (p) -- (v1);
\draw[-] (p) -- (vj);
\draw[-] (p) -- (vm);
\end{tikzpicture}
\caption{Construction of the tournament $T_x$ from an input $x=(x^{(1)},\ldots,x^{(m)})$ to $\OR_m\circ\PrMaj_m$.
The orientations inside $L$ form a fixed regular tournament, so every vertex of $L$ has out-degree $m$ within $L$. The orientations inside $R$, and between $p$ and $R$, are fixed so that every $v_j$ has out-degree $m/2$ within $R\cup\{p\}$.
The orientations of the dashed edges (i.e., the edges between $v_j$ and $\{u_i,\bar u_i\}$) depend on $x^{(i)}_j$. Specifically, if $x^{(i)}_j=1$ then set $u_i\to v_j\to\bar u_i$, and if $x^{(i)}_j=0$ then set $\bar u_i\to v_j\to u_i$. Thus,
$\deg^+_{T_x}(u_i)=m+\wt(x^{(i)})$ and
$\deg^+_{T_x}(\bar u_i)=m+(m-\wt(x^{(i)}))$.
Consequently, $\Delta^+(T_x)=3m/2$ if $(\OR_m\circ\PrMaj_m)(x)=0$, and
$\Delta^+(T_x)=3m/2+1$ if $(\OR_m\circ\PrMaj_m)(x)=1$.}
\label{fig:mod-new}
\end{figure}

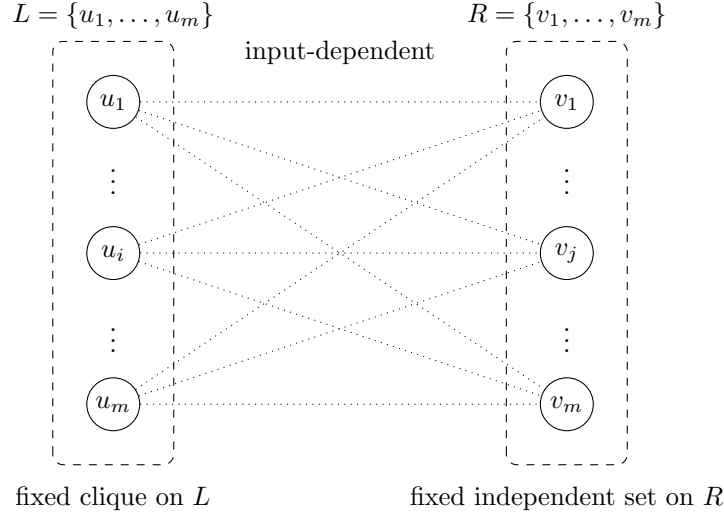
\begin{figure}[t]
\centering
\begin{tikzpicture}[
    scale=1,
    vertex/.style={circle,draw,minimum size=7mm,inner sep=0pt},
    every node/.style={font=\small}
]

\node[vertex] (u1) at (0,2.0) {$u_1$};
\node[vertex] (ui) at (0,0) {$u_i$};
\node[vertex] (um) at (0,-2.0) {$u_m$};

\node at (0,1.05) {$\vdots$};
\node at (0,-1.05) {$\vdots$};

\draw[dashed,rounded corners]
    (-0.8,-2.8) rectangle (0.8,2.8);

\node at (0,3.15)
    {$L=\{u_1,\ldots,u_m\}$};

\node[align=center] at (0,-3.25)
    {fixed clique on $L$};

\node[vertex] (v1) at (6,2.0) {$v_1$};
\node[vertex] (vj) at (6,0) {$v_j$};
\node[vertex] (vm) at (6,-2.0) {$v_m$};

\node at (6,1.05) {$\vdots$};
\node at (6,-1.05) {$\vdots$};

\draw[dashed,rounded corners]
    (5.2,-2.8) rectangle (6.8,2.8);

\node at (6,3.15)
    {$R=\{v_1,\ldots,v_m\}$};

\node[align=center] at (6,-3.25)
    {fixed independent set on $R$};

\draw[dotted] (u1) -- (v1);
\draw[dotted] (u1) -- (vj);
\draw[dotted] (u1) -- (vm);

\draw[dotted] (ui) -- (v1);
\draw[dotted] (ui) -- (vj);
\draw[dotted] (ui) -- (vm);

\draw[dotted] (um) -- (v1);
\draw[dotted] (um) -- (vj);
\draw[dotted] (um) -- (vm);

\node at (3,2.65) {input-dependent};

\end{tikzpicture}
\caption{Construction of the graph $G_x$ from an input
$x=(x^{(1)},\ldots,x^{(m)})$ to $\OR_m\circ\PrMaj_m$.
The set $L$ induces a fixed clique, while $R$ induces a fixed independent set.
The dotted edges (i.e., the edges between $u_i$ and $v_j$) depend on
$x^{(i)}_j$. Specifically, $\{u_i,v_j\}$ is an edge if and only if
$x^{(i)}_j=1$. Thus, $\deg_{G_x}(u_i)=m-1+\wt(x^{(i)})$. Consequently,
$\Delta(G_x)=3m/2-1$ if $(\OR_m\circ\PrMaj_m)(x)=0$, and
$\Delta(G_x)=3m/2$ if $(\OR_m\circ\PrMaj_m)(x)=1$.}
\label{fig:max-degree-new}
\end{figure}

Our constructions ensure that each query to an edge of $G_x$ or $T_x$ can be simulated using at most one query to the input strings. After finding a maximizing vertex, we compute its degree or out-degree using $O(m)$ additional queries, thereby evaluating $\OR_m\circ\PrMaj_m$. Consequently, 
\[
\Q(\umod_{2m})\geq\Q(\OR_m\circ\PrMaj_m)-O(m), \qquad \Q(\dmod_{3m+1})\geq\Q(\OR_m\circ\PrMaj_m)-O(m).
\]
See \Cref{lem:ormaj_to_maxdeg,lem:ormaj_to_maxoutdeg} for the precise statements.

By \Cref{thm:Q_composition}, we have $\Q(\OR_m\circ\PrMaj_m)=\Omega(\Q(\OR_m)\Q(\PrMaj_m))$. Combining the lower bounds $\Q(\OR_m)=\Omega(\sqrt{m})$ and $\Q(\PrMaj_m)=\Omega(m)$ from \Cref{thm:or_lower_bound,thm:prmaj_lower_bound} gives $\Q(\OR_m\circ\PrMaj_m)=\Omega(m^{3/2})$. The reductions therefore yield the lower bounds in \Cref{thm:max-degree,thm:mod}, with routine padding extending them to arbitrary numbers of vertices.

In particular, \Cref{thm:mod} resolves a conjecture of Mande, Paraashar, and Saurabh~\cite{MPS23}, who asked whether the $O(n^{3/2})$ quantum query upper bound for finding a maximum-out-degree vertex in a tournament is tight.

We next outline the upper and lower bounds in \Cref{thm:max-k-cut}, beginning with the upper bound.

For the upper bound in \Cref{thm:max-k-cut}, we first find the $k$ vertices of largest degree and let $A$ be the sum of their degrees. Repeatedly applying quantum maximum finding, with degrees computed exactly, uses $O(n^{3/2})$ queries for constant $k$. Every $k$-vertex cut has size at most $A$. Conversely, the cut defined by the selected vertices has size at least $A-k(k-1)$, since each edge between two selected vertices contributes twice to their degree sum but nothing to the cut. Thus, $A-k(k-1)\leq\mcut_{k,n}(G)\leq A$.

We then binary search over this interval to determine $\mcut_{k,n}(G)$. Each threshold test decides whether there exists a $k$-vertex set inducing a cut of at least the specified size, using $O(n^{2-1/(k+1)})$ queries via a standard quantum-walk algorithm with appropriately chosen parameters. For constant $k$, only constantly many tests are needed, so amplifying the subroutines to ensure bounded error overall preserves this query bound,\footnote{For general $1\leq k\leq n$, our algorithm uses $O\!\left(n^{2-1/(k+1)}\log(k+1)\log\log(k+2)\right)$ queries (see \Cref{lem:UpperBoundMaxCut} for a formal statement).} and thus
\[
\Q(\mcut_{k,n})=O(n^{2-1/(k+1)}).
\]

For the lower bound in \Cref{thm:max-k-cut}, fix an $N$-vertex graph $G$ and associate with every vertex $v\in V(G)$ a string $x_v\in\{0,1\}^N$, promised to have Hamming weight either $N/2$ or $N/2+1$. We regard $v$ as marked in the latter case. Determining whether the marked vertices contain a $k$-clique is precisely the composed problem $\mathsf{GC}^G_k\circ\PrMaj_N$.

We reduce $\mathsf{GC}^G_k\circ\PrMaj_N$ to $\mcut_{k,3N}$. The idea is to construct a graph in which marked vertices have slightly higher degrees than unmarked vertices, while every non-edge of $G$ within a selected set reduces its cut size. The constructed graph $H$ has three sets of $N$ vertices each, one of which is identified with $V(G)$. Inside $V(G)$, we place the edges of $\overline{G}$. We use one auxiliary set to make the contribution of the fixed edges incident to every $v\in V(G)$ equal to $N$, and the other to encode $x_v$. Thus, $\deg_H(v)=N+\wt(x_v)$ for every $v\in V(G)$. Hence, for every $k$-set $S\subseteq V(G)$,
\[
|\partial_H(S)|=\sum_{v\in S}\deg_H(v)-2|E(\overline{G}[S])|=\frac{3kN}{2}+\bigl|\{v\in S \mid \wt(x_v)=N/2+1\}\bigr|-2|E(\overline{G}[S])|.
\]

It follows that $|\partial_H(S)|=3kN/2+k$ if and only if every vertex of $S$ is marked and $S$ is a clique in $G$. Our construction also ensures that every auxiliary vertex has degree at most $N$, whereas every vertex in $V(G)$ has degree at most $3N/2+1$. Thus, any $k$-set containing an auxiliary vertex has cut size strictly below $3kN/2+k$.

Our construction ensures that every query to the adjacency matrix of $H$ can be simulated using at most one query to the strings $x_v$, giving a query-preserving reduction from $\mathsf{GC}^G_k\circ\PrMaj_N$ to $\mcut_{k,3N}$. Combining this reduction with \Cref{thm:Q_composition,thm:prmaj_lower_bound}, we obtain
\[
\Q(\mcut_{k,3N})=\Omega\!\left(N\cdot\Q(\mathsf{GC}^G_k)\right).
\]

\section{Preliminaries}\label{sec:prelims}

For a positive integer $n$, we write $[n]=\cbra{1,2,\dots,n}$. We write $p(n)=O(q(n))$, $p(n)=\Omega(q(n))$, and $p(n)=\Theta(q(n))$ to mean $p(n)\in O(q(n))$, $p(n)\in\Omega(q(n))$, and $p(n)\in\Theta(q(n))$, respectively. We denote the Hamming weight of a binary string $x$ by $\mathsf{wt}(x)$. All logarithms are taken base $2$. All graphs considered in this work are simple, with no self-loops or multiple edges. Unless stated otherwise, graphs on $n$ vertices have vertex set $[n]$. For $S\subseteq V(G)$, we write $G[S]$ for the subgraph of $G$ induced by $S$. For a finite set $S$, we write $\binom{S}{k}$ for the family of its $k$-element subsets. We use $\sqcup$ to denote disjoint union.

We identify an undirected graph $G$ on $n$ vertices with a Boolean string in $\zone^{\binom{n}{2}}$, indexed by unordered pairs $\{i,j\}\in\binom{[n]}{2}$, where the bit at index $\{i,j\}$ is $1$ if and only if $\{i,j\}\in E(G)$. We write $\deg_G(v)$ for the degree of a vertex $v$ and $\Delta(G)$ for the maximum degree of $G$, dropping the subscript $G$ when clear from context. For $S\subseteq V(G)$, we write $\partial_G(S)$ for the set of edges with exactly one endpoint in $S$. The complement $\overline{G}$ has vertex set $V(G)$ and edge set $E(\overline{G})=\binom{V(G)}{2}\setminus E(G)$. We write $K_n$ for the complete graph on $n$ vertices.

A tournament $T$ is an orientation of a complete graph. We identify a tournament on $n$ vertices with a Boolean string in $\zone^{\binom{n}{2}}$, indexed by unordered pairs $\{i,j\}\in\binom{[n]}{2}$. For $i<j$, the bit indexed by $\{i,j\}$ is $1$ if and only if the edge is oriented from $i$ to $j$. We write $u\to v$ when the edge between $u$ and $v$ is oriented from $u$ to $v$. We write $\odeg_T(v)$ and $\ideg_T(v)$ for the out-degree and in-degree of a vertex $v$, respectively, and $\Delta^+(T)$ for the maximum out-degree of $T$. For $S\subseteq V(T)$, we write $\partial_T^+(S)$ for the set of directed edges from $S$ to $V(T)\setminus S$. A tournament in which every vertex has the same out-degree is called a \emph{regular} tournament.

\subsection{Functions and relations of interest}

\begin{defi}[Maximum degree]\label{def:umod}
    Viewing each $G \in \zone^{\binom{n}{2}}$ as an undirected graph on $[n]$, define the relation $\umod_n \subseteq \zone^{\binom{n}{2}} \times [n]$ by $(G,v) \in \umod_n \iff \deg_G(v) = \max_{u \in [n]} \deg_G(u)$.
\end{defi}

\begin{defi}[Maximum out-degree]\label{def:dmod}
    Viewing each $T \in \zone^{\binom{n}{2}}$ as a tournament on $[n]$, define the relation $\dmod_n \subseteq \zone^{\binom{n}{2}} \times [n]$ by $(T,v) \in \dmod_n \iff \deg_T^+(v) = \max_{u \in [n]} \deg_T^+(u)$.
\end{defi}

Note that the relation defined in \Cref{def:dmod} is equivalent to the relation $\mathsf{MOD}_n$ defined in \cite[Page~3]{MPS23}.

\begin{defi}[$\mcut_{k,n}$]\label{def:maxcut}
    Let $1 \leq k \leq n$. Viewing each $G \in \zone^{\binom{n}{2}}$ as an undirected graph on $[n]$, define $\mcut_{k,n}(G) := \max_{S \subseteq [n], |S|=k} |\partial_G(S)|$.
\end{defi}

\begin{defi}[OR]\label{def:or}
    Let $\mathsf{OR}_n:\{0,1\}^n\rightarrow \{0,1\}$ be a Boolean function defined as $\mathsf{OR}_n(x) = 1 \iff \mathsf{wt}(x) \geq 1$.
\end{defi}

\begin{theorem}[{\cite[Theorem~3.3]{Bennett_Bernstein_Brassard_Vazirani_1997}}]\label{thm:or_lower_bound}
    For every positive integer $n$, $\Q(\mathsf{OR}_n) = \Omega(\sqrt{n})$.
\end{theorem}

\begin{defi}[$k$-Graph Collision]\label{def:kGC}
Let $G=(V,E)$ be a fixed undirected graph and let $1\leq k\leq |V|$. Define $\mathsf{GC}^G_k:\zone^V\to\zone$ by $\mathsf{GC}^G_k(x)=1\iff\exists S\subseteq V,\ |S|=k,\ \binom{S}{2}\subseteq E$  and $x_v=1$ for every $v\in S$.
When $k=2$, we simply refer to this as the \emph{Graph Collision} problem.
\end{defi}

\begin{defi}[Promise Majority]\label{def:prmaj}
Let $n$ be a positive even integer, and let $D=\{x\in\zone^n:\wt(x)\in\{n/2,n/2+1\}\}$. Define $\PrMaj_n:D\to\zone$ by $\PrMaj_n(x)=1\iff\wt(x)=\frac n2+1$.
\end{defi}

\subsection{Quantum query complexity}

Let $f:D\to Y$, where $D\subseteq\zone^n$ and $Y$ is a finite set. A quantum query algorithm for $f$ acts on a finite-dimensional Hilbert space $\mathcal H$ and alternates input-independent unitaries with applications of the query oracle
\[
    O_x\ket{i,b,z}
    =
    \ket{i,b\oplus x_i,z},
\]
where $i\in[n]$, $b\in\zone$, and $z$ denotes the workspace. Thus, after $T$ queries, the final state has the form
\[
    \ket{\psi_T}
    =
    U_T O_x U_{T-1} O_x \cdots U_1 O_x U_0\ket{\psi_0},
\]
where the unitaries $U_0,\ldots,U_T$ are independent of $x$. The algorithm then measures its output register in the computational basis and outputs the observed outcome.
The algorithm computes $f$ with bounded error if, for every input $x\in D$, it outputs $f(x)$ with probability at least $2/3$. The quantum query complexity $\Q(f)$ is the minimum number of queries used by any bounded-error quantum algorithm computing $f$. For a relation $R\subseteq D\times Y$, we similarly define $\Q(R)$ by requiring that, on every input $x\in D$, the algorithm outputs some $y\in Y$ such that $(x,y)\in R$ with probability at least $2/3$. For a more detailed introduction to the quantum query model, we refer the reader to \cite{Hamoudi_2025}.

\begin{theorem}[{\cite[Lemma~5.2]{Lee_Mittal_Reichardt_Spalek_Szegedy_2011}}]\label{thm:Q_composition}
Let $f:\zone^m\to\zone$ be a Boolean function and let $g:D\to\zone$, where $D\subseteq\zone^n$, be a partial Boolean function. Then $\Q(f\circ g)=\Omega\bigl(\Q(f)\Q(g)\bigr)$.
\end{theorem}

\begin{theorem}[{\cite[Theorem~1]{DH96}}]\label{thm:maximum_finding}
Let $S$ be a nonempty finite set, and suppose we have query access to a function $f:S\to\mathbb{R}$. There is a bounded-error quantum algorithm that outputs an element $s\in S$ maximizing $f(s)$ using $O(\sqrt{|S|})$
queries to $f$.
\end{theorem}

\begin{fact}\label{fact:gc-lower-bound}
For all integers $1\leq k\leq n$, $\Q(\mathsf{GC}^{K_n}_k)=\Omega(\sqrt{n-k+1})$. In particular, for every constant $k\geq 1$, we have $\Q(\mathsf{GC}^{K_n}_k)=\Omega(\sqrt n)$.
\end{fact}
\begin{proof}
Since every $k$-vertex subset of $K_n$ is a clique, $\mathsf{GC}^{K_n}_k$ asks whether at least $k$ vertices are marked. Fixing $k-1$ vertices to be marked restricts this function to $\OR_{n-k+1}$, giving the claimed lower bound.
\end{proof}

\begin{theorem}[{\cite{Amb02}}]\label{thm:prmaj_lower_bound}
    For every positive even integer $n$, $\Q(\mathsf{PrMaj}_n) = \Omega(n)$.
\end{theorem}

Ambainis~\cite{Amb02} actually showed an $\Omega(n)$ lower bound for the $n$-bit Majority function. It is not hard to see that the same proof also shows the same lower bound for for $\mathsf{PrMaj}_n$.

\subsection{Quantum walk on Johnson graph}
In this work, we will be making use of the MNRS-framework for quantum walks, developed by Magniez, Nayak, Roland and Santha~\cite{Magniez_Nayak_Roland_Santha_2011} as an algorithmic tool. We recall their result here.

\begin{theorem}[Quantum walk search {\cite[Theorem~3]{Magniez_Nayak_Roland_Santha_2011}}]
\label{thm:mnrs}
Let $P$ be a reversible, ergodic Markov chain on a finite state space $X$, with stationary distribution $\pi$ and spectral gap $\delta>0$. Let $M\subseteq X$ be the set of marked states, and
suppose that a known parameter $\varepsilon>0$ satisfies
\[
    M\neq\varnothing
    \quad\Longrightarrow\quad
    \pi(M)=\sum_{x\in M}\pi(x)\ge\varepsilon.
\]
Associate each state $x\in X$ with a data structure $D(x)$. Define the following costs, measured consistently in oracle queries (or elementary operations):
\begin{itemize}
    \item \emph{Setup cost $\bm{S}$:} the cost of preparing $        \sum_{x\in X}\sqrt{\pi(x)}\,|x\rangle|D(x)\rangle$.

    \item \emph{Update cost $\bm{U}$:}
Write $|x\rangle_D:=|x\rangle|D(x)\rangle$, and let
\[
    P^*(y,x):=\frac{\pi(x)P(x,y)}{\pi(y)}
\]
be the transition matrix of the time-reversed chain.
Let $\bm{U}$ bound the cost of implementing each of
$\mathcal{U}_P$, $\mathcal{U}_{P^*}$, and their inverses, where
\[
    \mathcal{U}_P:
    |x\rangle_D|\overline{0}\rangle
    \longmapsto
    |x\rangle_D
    \sum_{y\in X}\sqrt{P(x,y)}\,|y\rangle_D, \qquad 
    \mathcal{U}_{P^*}:
    |\overline{0}\rangle|y\rangle_D
    \longmapsto
    \sum_{x\in X}\sqrt{P^*(y,x)}\,|x\rangle_D|y\rangle_D.
\]
Here $|\overline{0}\rangle$ denotes a fixed blank state
of the state-and-data register. For a reversible chain,
$P^*=P$.

    \item \emph{Checking cost $\bm{C}$:} the cost of implementing
    the reflection $|x\rangle|D(x)\rangle
        \longmapsto
        (-1)^{\mathbb{I}[x \in M]}|x\rangle|D(x)\rangle$,
    with auxiliary workspace returned to its initial state.
\end{itemize}
There is a quantum algorithm that, with probability at least $2/3$, finds a marked state if one exists, or reports that $M=\varnothing$, at cost
\[
    O\!\left(
        \bm{S}+\frac{1}{\sqrt{\varepsilon}}
        \left(\frac{1}{\sqrt{\delta}}\bm{U}+\bm{C}\right)
    \right).
\]
\end{theorem}

For integers $n,r \in \mathbb{N}$,  the Johnson graph $J(n,r)$ is defined as follows.
\begin{defi}[Johnson graph $J(n,r)$]
For integers $1\le r<n$, the Johnson graph $J(n,r)$ is the undirected $r(n-r)$-regular graph whose vertices are the $r$-element subsets of $[n]$. Two vertices $R$ and $T$ are adjacent if and only if $|R\cap T|=r-1$, equivalently, if $T$ can be obtained from $R$ by replacing exactly one element.    
\end{defi}

Note that, for $n\ge 3$ and $1\le r<n$, the Johnson graph $J(n,r)$ is
connected, non-bipartite, and regular. Hence its simple random
walk is ergodic and reversible with respect to the uniform
stationary distribution
\cite[Sections~1.3, 1.4, and~1.6]{LevinPeres2017}. More formally,

\begin{theorem}[Ergodicity and reversibility of the walk over Johnson graph $J(n,r)$]
\label{thm:ergodicityJohnsonGraph}
Let $n\ge 3$ and $1\le r<n$. The Markov chain on
$X=\binom{[n]}{r}$ with transition probabilities
\[
    P(R,T)=
    \begin{cases}
        \dfrac{1}{r(n-r)}, & \text{if }|R\cap T|=r-1,\\[2mm]
        0, & \text{otherwise},
    \end{cases}
\]
is irreducible, aperiodic, and reversible with respect to
the uniform stationary distribution $\pi(R)=1/\binom{n}{r}$. In particular, it is ergodic.
\end{theorem}

\section{Computing maximum out-degree in a tournament}\label{sec:mod}

In this section we prove \Cref{thm:mod}, restated below, which is our tight bound for the quantum query complexity of finding a maximum-out-degree vertex in a tournament.

\thmmod*

We first prove the following lemma, which gives an almost query-preserving reduction from $\OR_a\circ\PrMaj_a$ to $\dmod_{3a+1}$.

\begin{lemma}\label{lem:ormaj_to_maxoutdeg}
For every positive even integer $a$,
\[
    \Q(\dmod_{3a+1})
    \geq
    \Q(\OR_a\circ\PrMaj_a)-3a.
\]
\end{lemma}

The theorem immediately follows as we note below.

\begin{proof}[Proof of \Cref{thm:mod}]
For $n=3a+1$, where $a$ is even, we have
\[
\Q(\mathsf{MOD}_{3a+1})
\geq
\Q\!\left(\mathsf{OR}_a\circ\mathsf{PrMaj}_a\right)-3a
=
\Omega\!\left(
\Q(\mathsf{OR}_a)\Q(\mathsf{PrMaj}_a)
\right)-3a
=
\Omega(a^{3/2})
=
\Omega(n^{3/2}).
\]
Here the inequality follows from \Cref{lem:ormaj_to_maxoutdeg}, the first equality follows from \Cref{thm:Q_composition}, and the second equality follows from \Cref{thm:or_lower_bound,thm:prmaj_lower_bound}.

For arbitrary $n$, let $n'=3a+1\leq n$ be the largest such integer with $a$ even, and pad an $n'$-vertex instance by adding $n-n'$ new vertices and orienting every edge from an old vertex towards a new vertex. Since $n'=\Theta(n)$, this gives the same $\Omega(n^{3/2})$ lower bound.

For the matching upper bound, the out-degree of any specified vertex can be computed exactly using $n-1$ queries. Applying quantum maximum finding \Cref{thm:maximum_finding} over the $n$ vertices, with the out-degree-computation procedure as the evaluation subroutine, therefore finds a maximum-out-degree vertex using $O(\sqrt{n}\cdot n)=O(n^{3/2})$ queries.\qedhere
\end{proof}

We now prove \Cref{lem:ormaj_to_maxoutdeg}.
\begin{proof}[Proof of \Cref{lem:ormaj_to_maxoutdeg}]

We prove the lemma by explicitly stating the reduction sketched in \Cref{fig:mod-new}. Let $x = x^{(1)},\ldots,x^{(a)}$ be an input to
$\mathsf{OR}_a\circ\mathsf{PrMaj}_a$, where each $x^{(i)}\in\zone^a$ and
\[
\wt(x^{(i)})\in\left\{\frac a2,\frac a2+1\right\}
\qquad\text{for every }i\in[a].
\]
We construct a tournament $G_x$ on vertex set
\[
L\sqcup R
=
\{u_1,\bar u_1,\ldots,u_a,\bar u_a,p\}
\sqcup
\{v_1,\ldots,v_a\}.
\]

We first fix a regular tournament on $L$. Thus, every vertex in $L$ has out-degree exactly $a$ within $L$.
Next, consider a regular tournament on $R\cup\{q\}$, where $q$ is a dummy vertex, and delete $q$ and its incident edges. At this point $a/2$ vertices in $R$ have out-degree $a/2-1$ within $R$, and $a/2$ vertices in $R$ have out-degree $a/2$ within $R$. For every $v_j\in R$, orient the edge between $p$ and $v_j$ so that $v_j$ has out-degree exactly $a/2$ within $R\cup\{p\}$. It follows that $p$ beats exactly $a/2$ vertices of $R$, and hence
\[
\deg^+_{G_x}(p)
=
a+\frac a2
=
\frac{3a}{2}.
\]
It remains to orient the edges between $R$ and the pairs $\{u_i,\bar u_i\}$. For every $i,j\in[a]$, we orient
\[
\begin{cases}
u_i\to v_j\to\bar u_i, & \text{if } x^{(i)}_j=1,\\
\bar u_i\to v_j\to u_i, & \text{if } x^{(i)}_j=0.
\end{cases}
\]
In particular, every $v_j$ beats exactly one vertex from each pair $\{u_i,\bar u_i\}$ regardless of the input. Therefore,
\[
\deg^+_{G_x}(v_j)
=
a+\frac a2
=
\frac{3a}{2}
\qquad\text{for every }j\in[a].
\]

For every $i\in[a]$, we have
\[
\deg^+_{G_x}(u_i)
=
a+\wt(x^{(i)}),
\qquad
\deg^+_{G_x}(\bar u_i)
=
a+(a-\wt(x^{(i)})).
\]
Hence, if $\wt(x^{(i)})=a/2$, both $u_i$ and $\bar u_i$ have out-degree $3a/2$, whereas if $\wt(x^{(i)})=a/2+1$, then $u_i$ and $\bar u_i$ have out-degrees $3a/2+1$ and $3a/2-1$, respectively.

Consequently,
\[
\Delta^+(G_x)
=
\begin{cases}
\frac{3a}{2}
& \text{if } \mathsf{PrMaj}_a(x^{(i)})=0
  \text{ for every }i\in[a],\\[2mm]
\frac{3a}{2}+1
& \text{otherwise}.
\end{cases}
\]
Thus,
\begin{equation}\label{eq:delta-out-orprmaj}
\Delta^+(G_x)=\frac{3a}{2}+1
\iff
\mathsf{OR}_a\!\left(
\mathsf{PrMaj}_a(x^{(1)}),\ldots,
\mathsf{PrMaj}_a(x^{(a)})
\right)=1.
\end{equation}

Given an input $x=(x^{(1)},\ldots,x^{(a)})$ to
$\mathsf{OR}_a\circ\mathsf{PrMaj}_a$, we can simulate queries to the tournament $G_x$ as follows. All fixed edge orientations can be answered without querying $x$. A query to the orientation of an edge between $v_j$ and either $u_i$ or $\bar u_i$ can be answered using a single query to $x^{(i)}_j$.

Hence, any $T$-query quantum algorithm for $\mathsf{MOD}_{3a+1}$, when run on $G_x$, can be simulated using $T$ queries to $x$. Let $v$ be the maximum-out-degree vertex output by this algorithm. We then compute $\deg^+_{G_x}(v)$ using at most $3a$ additional queries. We output $1$ if $\deg^+_{G_x}(v)=3a/2+1$, and output $0$ otherwise. By \Cref{eq:delta-out-orprmaj}, this gives a bounded-error quantum algorithm for
$\mathsf{OR}_a\circ\mathsf{PrMaj}_a$
using at most $T+3a$ queries.

Therefore,
\[
    \Q\!\left(\mathsf{OR}_a\circ\mathsf{PrMaj}_a\right) \leq \Q(\mathsf{MOD}_{3a+1})+3a.\qedhere
\]
\end{proof}

\section{Computing the maximum degree of an undirected graph}\label{sec:md}

In this section we prove \Cref{thm:max-degree}, restated below, which is our tight bound for the quantum query complexity of finding a maximum-degree vertex in an undirected graph.

\thmmaxdegree*

The lower bound is implicit in the proof of \cite[Theorem~37]{Apers_Lee_21}, whose hard instances can be transformed into unweighted maximum-degree instances by a simple query-preserving reduction. We give a direct proof here for completeness. The proof strategy is the same as in Section~3, but the reduction is simpler. We first prove the following almost query-preserving reduction from $\OR_{n/2}\circ\PrMaj_{n/2}$ to $\umod_n$.
\begin{lemma}\label{lem:ormaj_to_maxdeg}
For every positive integer $n$ divisible by 4,
\[
\Q(\mathsf{MD}_n)
\geq 
\Q\!\left(\mathsf{OR}_{n/2}\circ\PrMaj_{n/2}\right) - n + 1.
\]
\end{lemma}

The theorem immediately follows as we note below.

\begin{proof}[Proof of Theorem~\ref{thm:max-degree}]
First suppose that $n$ is divisible by 4.
We have
\[
\Q(\mathsf{MD}_n)
\geq
\Q\!\left(\mathsf{OR}_{n/2}\circ\mathsf{PrMaj}_{n/2}\right) - n + 1
=
\Omega\!\left(
\Q(\mathsf{OR}_{n/2})\Q(\mathsf{PrMaj}_{n/2})
\right) - n + 1
=
\Omega(n^{3/2}).
\]
Here the inequality follows from \Cref{lem:ormaj_to_maxdeg}, the first equality follows from \Cref{thm:Q_composition}, and the last equality follows from \Cref{thm:or_lower_bound,thm:prmaj_lower_bound}.

It remains to prove the lower bound for an arbitrary positive integer $n$. For arbitrary $n\geq 4$, let $n'=4\lfloor n/4\rfloor$. Pad an $n'$-vertex instance with $n-n'$ isolated vertices. Any maximum-degree vertex of the padded graph that belongs to the original graph remains a valid output for the original instance. If the algorithm outputs an added vertex, then, whenever its output is correct, the original graph is edgeless, so we may instead output any original vertex. Thus the lower bound extends to arbitrary $n$, since $n'=\Theta(n)$.

For the matching upper bound, the degree of any specified vertex can be computed exactly using $n-1$ queries. Applying quantum maximum finding \Cref{thm:maximum_finding} over the $n$ vertices, with the degree-computation procedure as the evaluation subroutine, therefore finds a maximum-degree vertex using $O(\sqrt{n}\cdot n)=O(n^{3/2})$ queries.\qedhere
\end{proof}

We now prove \Cref{lem:ormaj_to_maxdeg}.
\begin{proof}[Proof of \Cref{lem:ormaj_to_maxdeg}]

We prove the lemma by explicitly stating the reduction sketched in \Cref{fig:max-degree-new}. Let $x^{(1)},\ldots,x^{(n/2)}$ be an input to
$\mathsf{OR}_{n/2}\circ\mathsf{PrMaj}_{n/2}$, where each $x^{(i)} \in \zone^{n/2}$ and
\[
\wt(x^{(i)})\in\left\{\frac n4,\frac n4+1\right\}
\qquad\text{for every }i\in[n/2].
\]
We construct a graph $G_x$ on vertex set
\[
    L\sqcup R = \{u_1,\ldots,u_{n/2}\} \sqcup \{v_1,\ldots,v_{n/2}\}.
\]
Add edges to $G_x$ such that $L$ becomes a clique and $R$ remains an independent set. For every
$i,j\in[n/2]$, we include the edge $\{u_i,v_j\}$ if and only if
$x^{(i)}_j=1$.

Thus, for every $i\in[n/2]$, $\deg_{G_x}(u_i) = \frac n2-1+\wt(x^{(i)}) \in \left\{\frac{3n}{4}-1,\frac{3n}{4}\right\}$.
On the other hand, every vertex in $R$ has degree at most $n/2$.
Consequently, the maximum degree of a vertex in $G_x$ is
\[
    \Delta(G_x) = \begin{cases}
        \frac{3n}{4}-1
        & \text{if } \mathsf{PrMaj}_{n/2}(x^{(i)})=0
          \text{ for every } i\in[n/2],\\[2mm]
        \frac{3n}{4}
        & \text{otherwise}.
    \end{cases}
\]
Thus,
\begin{equation}\label{eq:delta-orprmaj}
\Delta(G_x)=3n/4
\iff
\mathsf{OR}_{n/2}\!\left(
\mathsf{PrMaj}_{n/2}(x^{(1)}),\ldots,
\mathsf{PrMaj}_{n/2}(x^{(n/2)})
\right)=1.
\end{equation}

Given an input
$x=(x^{(1)},\ldots,x^{(n/2)})$
to $\mathsf{OR}_{n/2}\circ\mathsf{PrMaj}_{n/2}$, we can simulate queries to the adjacency matrix of $G_x$ as follows: queries to edges inside $L$ or inside $R$ can be answered without querying $x$, while a query to an edge $\{u_i,v_j\}$ can be answered using a single query to $x^{(i)}_j$.

Hence, any $T$-query quantum algorithm for $\mathsf{MD}_n$, when run on $G_x$, can be simulated using $T$ queries to $x$. Let $v$ be the maximum-degree vertex output by this algorithm. We then compute $\deg_{G_x}(v)$ using at most $n-1$ additional queries. We output $1$ if $\deg_{G_x}(v)=3n/4$, and output $0$ otherwise. By \Cref{eq:delta-orprmaj}, this gives a bounded-error quantum algorithm for
$\mathsf{OR}_{n/2}\circ\mathsf{PrMaj}_{n/2}$
using at most $T+n-1$ queries.

Therefore,
\[
    \Q\!\left(\mathsf{OR}_{n/2}\circ\mathsf{PrMaj}_{n/2}\right) \leq \Q(\mathsf{MD}_n)+n-1.\qedhere
\]
\end{proof}

\section{Generalizing the notion of maximum degree}

For $k=1$, $\mcut_{k,n}(G)$ is the maximum degree of $G$, so $\Q(\mcut_{1,n})=\Theta(n^{3/2})$ by \Cref{thm:max-degree}.
We now consider general $k$.

\thmmaxcut*

\subsection{Upper bound}
We first give a quantum algorithm for deciding whether $\mcut_{k,n}(G)$ is at least a given threshold, and then use it to compute $\mcut_{k,n}(G)$.

\begin{lemma}\label{lem:maxcut-threshold}
Let $1\leq k\leq n$ be integers. For every integer $t$, there is a bounded-error quantum algorithm that decides whether $\mcut_{k,n}(G)\geq t$ using $O(n^{2-\frac{1}{k+1}})$ queries to an $n$-vertex graph $G$.
\end{lemma}
\begin{proof}
Let $G=(V,E)$ be the input graph. We consider two cases depending on $k$.

First suppose that $1\leq k\leq \log n$. We apply the quantum-walk algorithm from \Cref{thm:mnrs} to the Johnson graph $J(n,r)$, whose vertices are the $r$-element subsets of $V$ with two subsets adjacent if they differ in exactly one element. For each $R\in\binom{V}{r}$, the algorithm stores the neighborhood in $G$ of every vertex in $R$ in the database.

We declare $R$ to be marked if there exists a set $S\subseteq R$ with $|S|=k$ such that $|\partial_G(S)|\geq t$.
Since all adjacency information concerning the vertices in $R$ is stored, whether $R$ is marked can be determined without making any additional queries. Then, the query costs of the setup, the update and the checking steps are as follows:
\begin{itemize}
    \item The setup cost is $\bm{S}=O(rn)$, since we query all edges of the input graph incident to the $r$ vertices in $R$.
    \item The update cost is $\bm{U}=O(n)$, since replacing one vertex of $R$ requires uncomputing the stored adjacency information of the removed vertex and computing the adjacency information of the newly added vertex.
    \item The checking cost is $\bm{C}=0$, since whether $R$ is marked can be determined from the stored adjacency information.
\end{itemize}
The spectral gap of $J(n,r)$ is $\delta=\frac{n}{r(n-r)}>\frac1r$ \cite[Section~12.3.2]{brouwer2011spectra}. Moreover, if there exists a set $S\subseteq V$ with $|S|=k$ and
$|\partial_G(S)|\geq t$, then every $R\in\binom{V}{r}$ containing $S$ is marked. Hence the number of marked vertices is at least
$\binom{n-k}{r-k}$.
Using \Cref{thm:ergodicityJohnsonGraph}, the stationary distribution of $J(n,r)$ satisfies $\pi(R) = 1/\binom{n}{r}$ for all vertices $R$. Hence, using $M$ to denote the set of marked vertices, 
\[
\sum_{x \in M}\pi(x) \geq \frac{\binom{n-k}{r-k}}{\binom nr}
= \frac{\binom rk}{\binom nk}.
\]
Thus setting $\varepsilon$ equal to the RHS above gives a valid setting for \Cref{thm:mnrs}. Using \Cref{thm:mnrs} we get
\[
O\left(rn+n\sqrt{r}\sqrt{\frac{\binom{n}{k}}{\binom{r}{k}}}\right)
=
O\left(rn+n\sqrt{r}\left(\frac{en}{r}\right)^{k/2}\right),
\]
where we used $\binom{n}{k}\leq (en/k)^k$ and $\binom{r}{k}\geq (r/k)^k$. Set $a=(en)^{k/(k+1)}$. If $\lceil a\rceil\geq n$, we query the entire graph using $O(n^2)=O(na)$ queries. Otherwise, take $r=\lceil a\rceil$, so that $k\leq r<n$. Since $a^{k+1}=(en)^k$ and $r\geq a$, 
\[
rn+n\sqrt{r}\left(\frac{en}{r}\right)^{k/2}
=
rn+n(en)^{k/2}r^{(1-k)/2}
\leq rn+na
=
O(na)
=
O\left(n^{2-\frac1{k+1}}\right).
\]
For $k>\log n$, we simply query the entire graph using $\binom{n}{2}=O(n^2)$ queries. Since $n^2\leq e\,n^{2-1/(k+1)}=O\left(n^{2-\frac1{k+1}}\right)$ for $k>\log n$, this gives the claimed bound in all cases.
\end{proof}

\begin{lemma}\label{lem:maxcut-degree-sum}
Let $1\leq k\leq n$ and $G$ be an $n$-vertex graph. Let $D$ consist of $k$ vertices of highest degree in $G$, and write $A=\sum_{v\in D}\deg_G(v)$. Then
\[
    A-k(k-1)\leq \mcut_{k,n}(G)\leq A.
\]
\end{lemma}

\begin{proof}
For every $k$-vertex set $S$, we have
$|\partial_G(S)|\leq\sum_{v\in S}\deg_G(v)\leq A$.
For the lower bound, choose $S=D$ and observe that $|\partial_G(D)|=A-2|E(G[D])|\geq A-k(k-1)$.
\end{proof}

\begin{lemma}\label{lem:maxcut-interval}
Let $1\leq k\leq \log n$. There is a quantum algorithm that, given query access to an $n$-vertex graph $G$, outputs an integer $A$ satisfying
\[
    A-k(k-1)\leq \mcut_{k,n}(G)\leq A
\]
with probability at least $5/6$, using $O(kn^{3/2}\log (k+1))$ queries.
\end{lemma}
\begin{proof}
For $1\leq k \leq \log n$, find the $k$ highest-degree vertices by iteratively applying the maximum-finding procedure from the proof of \Cref{thm:max-degree}, excluding previously selected vertices but always computing degrees in the original graph $G$. Reduce the error to $1/(6k)$ per invocation, incurring an $O(\log (k+1))$ factor. Compute the degrees of the selected vertices using at most $k(n-1)$ additional queries, and output their sum. Correctness follows from a union bound and \Cref{lem:maxcut-degree-sum}.
\end{proof}

\begin{lemma}\label{lem:UpperBoundMaxCut}
Let $1\leq k\leq n$. Then
\[
    \Q(\mcut_{k,n})
    = O\left(n^{2-1/(k+1)}\log (k+1)\log\log (k+2) \right).
\]
\end{lemma}
\begin{proof}
For $1\leq k \leq \log n$, using \Cref{lem:maxcut-interval} we output the value of $A$ (with error at most $1/6$) and then use binary search to find the largest value of $t \in [A-k(k-1), A]$ such that $\mcut_{k,n}(G)\geq t$. Binary search requires at most $2\log(k+1)$ threshold tests. We reduce the error of each test to $1/(12\log(k+1))$, incurring an $O(\log\log(k+2))$ factor in its query cost. Conditioned on $A$ satisfying the inequalities in \Cref{lem:maxcut-degree-sum}, a union bound shows that binary search returns the correct value with probability at least $5/6$. Combining these two subroutines gives us a $1/3$-bounded error algorithm that ensures
\begin{align*}
    \Q(\mcut_{k,n})
    &=O\left(kn^{3/2}\log(k+1)
    +n^{2-1/(k+1)}\log(k+1)\log\log(k+2)\right)\\
    &=O\left(n^{2-1/(k+1)}\log(k+1)\log\log(k+2)\right).
\end{align*}
The last equality holds since $1\leq k \leq \log n$ and the first term is absorbed into the second term. 

For $k > \log n$, we query $E(G)$ using ${n \choose 2 }=O(n^{2-1/(k+1)})$ queries. This concludes the proof.
\end{proof}

\subsection{Lower bounds via graph collision}

Recall the definitions of $\mathsf{GC}^G_k$ and $\PrMaj_n$ from \Cref{def:kGC,def:prmaj}. We show a query-preserving reduction from $\mathsf{GC}^G_k\circ\PrMaj_n$ to $\mcut_{k,3n}$.

\begin{lemma}\label{thm:GClbMaxCut}
For every even positive integer $n$, integer $1 \leq k \leq n$, and every $n$-vertex graph $G=(V,E)$,
\[
    \Q(\mcut_{k,3n})
    =
    \Omega\bigl(n\cdot \Q(\mathsf{GC}^G_k)\bigr).
\]
\end{lemma}

\begin{proof}
Let $D=\{x\in\{0,1\}^n \mid \wt(x)\in\{n/2,n/2+1\}\}$.
An input to $\mathsf{GC}^G_k\circ\PrMaj_n$ is $x=(x_v)_{v\in V}\in D^V$, where $v\in V$ is marked if and only if $\wt(x_v)=n/2+1$, i.e., $\PrMaj_n(x_v) = 1$.

We construct a graph $H_x$ on three sets of vertices
$V,L=\{u_1,\ldots,u_n\},R=\{w_1,\ldots,w_n\}$,
each of size $n$.

Recall that $\bar G=(V,\bar E)$ denotes the complement of $G$. Inside $V$, we place exactly the edges of $\bar G$. For every $i,j\in[n]$, add the edge $\{v_i,u_j\}$ if and only if either $i=j$ or $\{v_i,v_j\}\in E(G)$. Thus, for every $v_i\in V$, the number of its neighbors in $V\sqcup L$ is
\[
    \deg_{\bar G}(v_i)+\deg_G(v_i)+1=n.
\]
Finally, for every $i,j\in[n]$, add the edge $\{v_i,w_j\}$ if and only if $x_{v_i}[j]=1$. Therefore,
\[
    \deg_{H_x}(v_i)
    =
    n+\wt(x_{v_i})
    \in
    \left\{\frac{3n}{2},\frac{3n}{2}+1\right\}.
\]
See \Cref{fig:gc-to-maxcut} for an illustration of the construction.

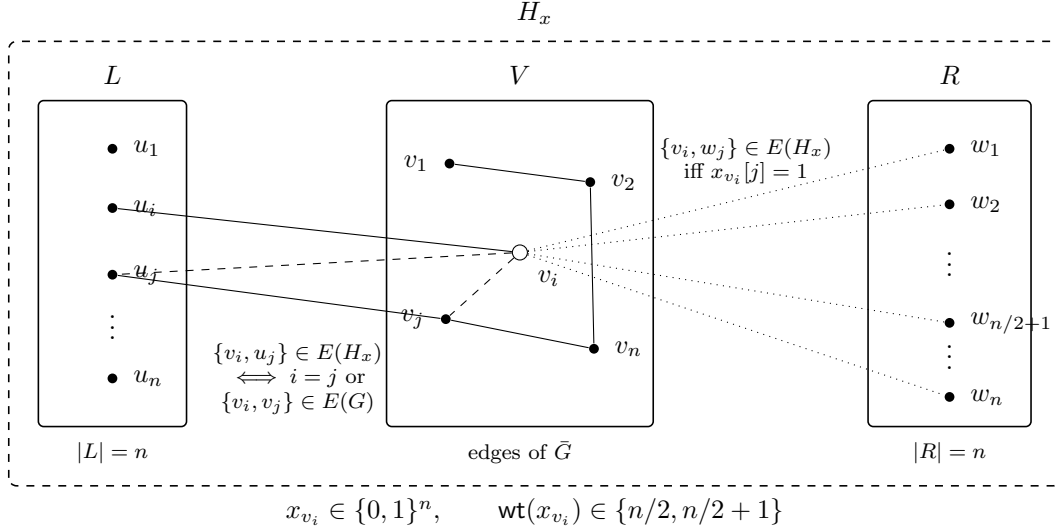
\begin{figure}[ht]
\centering
\begin{tikzpicture}[
scale=0.98,
every node/.style={font=\small},
vertex/.style={circle, fill=black, inner sep=1.3pt},
special/.style={circle, draw, fill=white, inner sep=2pt},
component/.style={draw, rounded corners=2pt, line width=0.6pt}
]

\draw[component] (-6.5,-2.2) rectangle (-4.5,2.2);
\node at (-5.5,2.55) {$L$};

\node[vertex] (u1) at (-5.5,1.55) {};
\node[right=2pt] at (u1.east) {$u_1$};

\node[vertex] (ui) at (-5.5,0.75) {};
\node[right=2pt] at (ui.east) {$u_i$};

\node[vertex] (uj) at (-5.5,-0.15) {};
\node[right=2pt] at (uj.east) {$u_j$};

\node at (-5.5,-0.75) {$\vdots$};

\node[vertex] (un) at (-5.5,-1.55) {};
\node[right=2pt] at (un.east) {$u_n$};

\node at (-5.5,-2.55) {\scriptsize $|L|=n$};

\draw[component] (-1.8,-2.2) rectangle (1.8,2.2);
\node at (0,2.55) {$V$};

\node[special] (vi) at (0,0.15) {};
\node[below right=1pt] at (vi.south east) {$v_i$};

\node[vertex] (v1) at (-0.95,1.35) {};
\node[left=2pt] at (v1.west) {$v_1$};

\node[vertex] (v2) at (0.95,1.1) {};
\node[right=2pt] at (v2.east) {$v_2$};

\node[vertex] (vjv) at (-1.0,-0.75) {};
\node[left=2pt] at (vjv.west) {$v_j$};

\node[vertex] (vn) at (1.0,-1.15) {};
\node[right=2pt] at (vn.east) {$v_n$};

\draw (v1) -- (v2);
\draw (v2) -- (vn);
\draw (vn) -- (vjv);

\node at (0,-2.55) {\scriptsize edges of $\bar G$};

\draw[component] (4.7,-2.2) rectangle (6.9,2.2);
\node at (5.8,2.55) {$R$};

\node[vertex] (w1) at (5.8,1.55) {};
\node[right=2pt] at (w1.east) {$w_1$};

\node[vertex] (w2) at (5.8,0.8) {};
\node[right=2pt] at (w2.east) {$w_2$};

\node at (5.8,0.1) {$\vdots$};

\node[vertex] (wh) at (5.8,-0.8) {};
\node[right=2pt] at (wh.east) {$w_{n/2+1}$};

\node at (5.8,-1.15) {$\vdots$};

\node[vertex] (wn) at (5.8,-1.8) {};
\node[right=2pt] at (wn.east) {$w_n$};

\node at (5.8,-2.55) {\scriptsize $|R|=n$};

\draw (vi) -- (ui);
\draw[dashed] (vi) -- (uj);
\draw[dashed] (vi) -- (vjv);

\draw[] (vjv) -- (uj);

\node at (-3.0,-1.22) {\scriptsize $\{v_i,u_j\}\in E(H_x)$};
\node at (-3.0,-1.55) {\scriptsize $\iff i=j$ or};
\node at (-3.0,-1.88) {\scriptsize $\{v_i,v_j\}\in E(G)$};

\draw[dotted] (vi) -- (w1);
\draw[dotted] (vi) -- (w2);
\draw[dotted] (vi) -- (wh);
\draw[dotted] (vi) -- (wn);

\node at (3.05,1.55) {\scriptsize $\{v_i,w_j\}\in E(H_x)$};
\node at (3.05,1.22) {\scriptsize iff $x_{v_i}[j]=1$};

\draw[dashed, rounded corners=3pt, line width=0.6pt]
    (-6.9,-3.0) rectangle (7.3,3.0);

\node at (0.2,3.38) {$H_x$};
\node at (0.2,-3.35) {$x_{v_i}\in\{0,1\}^n,\qquad \wt(x_{v_i})\in\{n/2,n/2+1\}$};

\end{tikzpicture}

\caption{
Construction of the graph $H_x$ from an input $x=(x_v)_{v\in V}$ to $\mathsf{GC}^G_k\circ\PrMaj_n$. The graph $H_x$ has vertex set $V\sqcup L\sqcup R$, where $L=\{u_1,\ldots,u_n\}$ and $R=\{w_1,\ldots,w_n\}$. Inside $V$, we place the complement graph $\bar G$. For every $i,j\in[n]$, we add the edge $\{v_i,u_j\}$ if and only if either $i=j$ or $\{v_i,v_j\}\in E(G)$. Finally, for every $i,j\in[n]$, we add the edge $\{v_i,w_j\}$ if and only if $x_{v_i}[j]=1$. Consequently, every vertex $v_i\in V$ has degree $n+\wt(x_{v_i})\in\{3n/2,3n/2+1\}$.}
\label{fig:gc-to-maxcut}
\end{figure}
We now characterize the maximum $k$-cut in $H_x$. Every vertex in $L\sqcup R$ has degree at most $n$, whereas every vertex in $V$ has degree at most $3n/2+1$. Hence, any $k$-set containing a vertex from $L\sqcup R$ induces a cut of size strictly less than $k(3n/2+1)$.
On the other hand, for every $S\subseteq V$ with $|S|=k$,
\[
    |\partial_{H_x}(S)|
    =
    \sum_{v\in S}\bigl(n+\wt(x_v)\bigr)-2|E(\bar G[S])|
    \leq
    k\left(\frac{3n}{2}+1\right).
\]
Moreover, equality holds if and only if $\wt(x_v)=n/2+1$ for every $v\in S$ and $E(\bar G[S])=\emptyset$. Equivalently, $\mcut_{k,3n}(H_x)=k(3n/2+1)$ if and only if $G$ contains a $k$-clique consisting entirely of marked vertices. Therefore,
\[
    (\mathsf{GC}^G_k\circ\PrMaj_n)(x)=1
    \iff
    \mcut_{k,3n}(H_x)\geq k\left( \frac{3n}{2}+1 \right).
\]
Finally, the reduction is query preserving. Every edge of $H_x$ except those between $V$ and $R$ is fixed and can therefore be answered without querying $x$. A query to an edge $\{v_i,w_j\}$ can be answered using one query to $x_{v_i}[j]$. Hence,
\[
    \Q(\mathsf{GC}^G_k\circ\PrMaj_n)
    =
    O(\Q(\mcut_{k,3n})).
\]
By \Cref{thm:Q_composition} and the fact that $\Q(\PrMaj_n)=\Theta(n)$, we obtain $\Q(\mcut_{k,3n})=\Omega\bigl(n\cdot\Q(\mathsf{GC}^G_k)\bigr)$, which completes the proof.\qedhere
\end{proof}

We now prove \Cref{thm:max-k-cut}.

\begin{proof}[Proof of \Cref{thm:max-k-cut}]
For constant $k$, the upper bound follows from \Cref{lem:UpperBoundMaxCut}. The second assertion is precisely \Cref{thm:GClbMaxCut}. We now prove the unconditional lower bound.
    Applying \Cref{thm:GClbMaxCut} with $G=K_n$ and using \Cref{fact:gc-lower-bound}, we obtain, for every constant $k\geq 1$,
\[
    \Q(\mcut_{k,3n})=\Omega\bigl(n\cdot\Q(\mathsf{GC}^{K_n}_k)\bigr)
    =\Omega(n^{3/2}).
\]
For arbitrary sufficiently large $n$, let $N=2\lfloor n/6\rfloor$. Apply the construction in \Cref{thm:GClbMaxCut} with $G=K_N$, and pad the resulting $3N$-vertex graph with $n-3N$ isolated vertices. Every vertex has degree at most $3N/2+1$, so no $k$-set containing an added vertex can attain the threshold $k(3N/2+1)$. Thus the padded graph attains this threshold if and only if the original constructed graph does, preserving the reduction. Since $N=\Theta(n)$, the $\Omega(n^{3/2})$ lower bound follows.
\end{proof}

\corunconditionalmaxcut*
\begin{proof}[Proof of \Cref{cor:max-k-cut-unconditional}]
Fix $k>\ell\geq 2$ and let $N=2\lfloor n/6\rfloor$.
By \cite[Theorem~1.2]{Gilani_Zhou_2026}, for all sufficiently large $N$, there is an $N$-vertex graph $G$ satisfying
\[
\Q(\mathsf{GC}^G_k)\geq N^{\ell(2k-\ell+1)/(2k(\ell+1))-o(1)}.
\]
Applying \Cref{thm:GClbMaxCut} gives
\[
\Q(\mcut_{k,3N})\geq N^{1+\ell(2k-\ell+1)/(2k(\ell+1))-o(1)}.
\]
Pad the constructed $3N$-vertex instances with $n-3N$ isolated vertices. Every vertex in the constructed graph has degree at most $3N/2+1$, so no $k$-set containing an added vertex can attain the reduction's threshold $k(3N/2+1)$. Thus padding preserves the reduction. Since $N=\Theta(n)$, the claimed lower bound follows for arbitrary sufficiently large $n$.

To optimize the exponent, write $f_k(\ell)=1+\ell(2k-\ell+1)/(2k(\ell+1))$. Then
\[
f_k(\ell+1)-f_k(\ell)
=\frac{2k+2-(\ell+1)(\ell+2)}{2k(\ell+1)(\ell+2)}.
\]
Hence an optimal choice is the smallest integer $\ell\geq 2$ satisfying $(\ell+1)(\ell+2)\geq 2k+2$, namely $\ell=\lceil(\sqrt{8k+9}-3)/2\rceil$. For $k\geq 3$, this choice satisfies $2\leq\ell<k$.
\end{proof}

\section*{Acknowledgments}

N.R.\ thanks Emanuel Juliano for helpful discussions. A.S.G received support from the National Science Foundation (grant 26-17356). S.P.\ and N.R.\ acknowledge support from the Dutch Ministry of Education, Culture, and Science through Gravitation project ``Challenges in Cyber Security - 024.006.037'' for this work.

\section*{AI Statement}
We used OpenAI’s ChatGPT (GPT-5.5 and GPT-5.6) in developing some of the proofs and in editing the exposition. The authors first reduced the maximum-out-degree lower-bound problem to a problem about row and column sums of a Boolean matrix, which we then supplied to ChatGPT. ChatGPT provided a lower-bound argument using the positive-weight adversary method. The authors subsequently recognized that the construction could be presented more simply as a reduction from $\OR \circ \PrMaj$, which is the approach used in this paper. The graph-collision-based lower bounds were also developed through iterative discussions with ChatGPT. The authors checked and refined the resulting arguments and take full responsibility for all results and proofs presented here.

\bibliography{bibo}

\end{document}